\documentclass[11pt]{article}
\pdfoutput=1
\usepackage[centertags]{amsmath}
\usepackage{amssymb}

\newtheorem{theorem}{Theorem}
\newtheorem{proposition}{Proposition}[section]
\newtheorem{definition}{Definition}[section]

\newtheorem{lemma}{Lemma}[section]

\newtheorem{problem}{Problem}

\providecommand{\expectation}[2]{\mathbb{E}_{#2}\left[#1\right]}
\providecommand{\probab}[2]{\mathbb{P}_{#2}\left\{#1\right\}}

\providecommand{\whp}{\textbf{whp}}

\providecommand{\rg}[1]{G(#1)}

\providecommand{\rig}[1]{G_{\mathcal{I}}(#1)}
\providecommand{\rsg}[1]{G_{\mathcal{S}}(#1)}

\providecommand{\rigs}{random intersection graphs}
\providecommand{\ktree}{$k$-tree}
\providecommand{\ktrees}{$k$-trees}

\providecommand{\binom}[2]{{#1\choose#2}}

\providecommand{\degree}[2]{{\textrm{deg}_{#1}(#2)}}

\providecommand{\qed}{\hfill $\blacksquare$}

\newenvironment{proof}[0]{\textit{Proof.} }{\hfill  \qed} 

\newenvironment{problem-env}[1]{
\vskip 3mm
\begin{problem} 
\textbf{#1}
}
{
\end{problem}
}

\newenvironment{neat-list}[3]{
  \vskip 1mm
  \begin{list}{}
  {\itemsep 2mm  \labelsep #1 \labelwidth #2 \leftmargin #3
   \partopsep 0mm \topsep 0mm \parsep 0mm  \itemindent 0mm}
}
{
  \end{list}
  \vskip 3mm
}

\begin{document}

\title{Treewidth of Erd\"{o}s-R\'{e}nyi Random Graphs, Random Intersection Graphs, and 
  Scale-Free Random Graphs
}

\author{Yong Gao \thanks{Supported in part by NSERC Discovery Grant RGPIN 327587-09} \\
    Department of Computer Science, \\
    Irving K. Barber School of Arts and Sciences \\
    University of British Columbia Okanagan, \\
    Kelowna, Canada V1V 1V7 \\
}   
\maketitle 

\begin{abstract}
We prove that the treewidth of an Erd\"{o}s-R\'{e}nyi random graph $\rg{n, m}$ is, with
high probability, greater than $\beta n$ for some constant $\beta > 0$ if 
the edge/vertex ratio $\frac{m}{n}$ is greater than 1.073. Our lower bound 
$\frac{m}{n} > 1.073$ improves the only previously-known lower bound established in \cite{kloks94}.  
We also study the treewidth of random graphs under two other random models for large-scale complex networks. In particular, our result on the treewidth of \rigs~strengths a previous observation in \cite{karonski99cpc} on the average case behavior of the \textit{gate matrix layout} problem. For scale-free random graphs based on the  Barab\'{a}si-Albert preferential-attachment model, our result shows that  if more than 12 vertices are attached to a new vertex, then the treewidth of the obtained network is linear in the size of the network with high probability.   
\end{abstract}

\section{Introduction}
\label{Introduction}
Treewidth plays an important role in characterizing the structural properties of a graph
and the complexity of a variety of algorithmic problems of practical importance
\cite{bodlaender93,kloks94}. When restricted to instances with bounded treewidth, 
many NP-hard problems are polynomially sovable. Dynamic programming algorithms based on the
tree-decomposition of graphs have found many applications in research field
such as computational biology and artificial intelligence \cite{dalmau02,dechter01}.

The theory of random graphs pioneered by the work of Erd\"{o}s and R\'{enyi} 
\cite{erdos60} deals with the probabilistic behavior of various graph properties such as the connectivity, the colorability, and the size of (connected) components
\cite{erdos60,bollobas01,achlioptas99sharp,friedgut99}.
Random intersection graphs and scale-free random graphs were proposed 
as more realistic models for large-scale complex networks arising in
real-world domains such as communication networks (Internet, WWW, Wireless and P2P networks), computational biology (protein networks), and sociology 
(social networks). It has been hoped that these new models will be able to capture    
the common features of these networks in a better way and in the mean time, are mathematically
approachable and algorithmically tractable \cite{cooper05,ferrante08tcs,gao09tcs,silvio09stoc}.  

As treewidth is one of the most important structural parameters used to capture the 
algorithmic tractability of computationally hard problems, it is interesting 
to see how large the treewidth of  a typical graph is in these random models. Of course, studying the probabilistic behavior of the treewidth of these random graphs is itself an interesting combinatorial problem. 
Except for a result in \cite{kloks94} establishing an lower bound  on the threshold of having a linear treewidth of the Erd\"{o}s-R\'{e}nyi random graph, we are not aware of any other work in the literature. In the paper, we study the treewidth of random graphs under 
the following three random models:
\begin{enumerate}
\item \textbf{The Erd\"{o}s-R\'{e}nyi model \cite{bollobas01,erdos60}}. An Erd\"{o}s-R\'{e}nyi random graph $\rg{n, m}$ is defined on $n$ vertices and contain $m$ edges selected from the 
$N = \binom{n}{2}$ potential edges uniformly at random and without replacement.
\item \textbf{The random intersection model \cite{karonski99cpc}}. A random intersection graph $\rig{n, m, p}$ on $n$ vertices is defined as follows. Let $M = \{1, 2, \cdots, m\}$ be a fixed universe of size $m$. Each vertex  $v$ is associated with a subset $S_v \subset M$ that is obtained by including each element in $M$ independently with probability $p$. These $S_v$'s are determined independently as well. There is an edge between a pair of vertices $u$ and $v$ if and only if $S_u \cap S_v \not= \emptyset$. 
\item \textbf{The Barab\'{a}si-Albert scale-free model \cite{albert02complex}}. 
A Barab\'{a}si-Albert random graph  $\rsg{n, m}$ on a set of $n$ vertices $\{v_1, \cdots, v_n\}$ is defined by a graph evolution process in which vertices are added to the graph one at a time.
In each step, the newly-added vertex is connected to
$m$ existing vertices selected according to 
the \textit{preferential attachment} mechanism, i.e. an existing vertex is selected with probability in proportion to its degree.          
\end{enumerate} 

We establish a lower bound  1.073 on the edge/vertex ratio $\frac{m}{n}$ 
above which an Erd\"{o}s-R\'{e}nyi random graph $\rg{n, m}$ has a treewidth 
linear to the number of vertices with high probability. Our lower bound
improves the previous one $\frac{m}{n} > 1.18$ in \cite{kloks94}. 
We obtain similar results on the behavior of the treewidth 
for the random intersection graph $\rig{n, m, p}$ and the  Barab\'{a}si-Albert scale-free random graph $\rsg{n, m}$. Our result on $\rig{n, m, p}$ complements an observation in \cite{karonski99cpc} on the average case behavior of the \textit{gate matrix layout} problem. 
Our result on the scale-free random graph $\rsg{n, m}$ shows that if more than 12 vertices are attached to a new vertex, then the treewidth of the obtained network is linear in the size of the network with high probability.             
Our results are summarized in the following theorems: 
 
\begin{theorem}
\label{thm-treewidth-bound}
Let $\rg{n, m}$ be an Erd\"{o}s-R\'{e}nyi random graph. 
For any $\frac{m}{n} \geq 1.073$, there is a constant
$\beta > 0$ such that
\begin{equation}
\label{treewidth-bound-eq}
  \lim\limits_{n\rightarrow \infty}\probab{tw(\rg{n, m}) > \beta n}{\rg{n, m}} = 1.
\end{equation}
\end{theorem}

\begin{theorem}
\label{thm-intersection-graph}
Let $\rig{n, m, p}$ be a random intersection graph with the universe 
$M = \{1, \cdots, m\}$ and $m = n^{\alpha}$. 
For any $p \geq \frac{2}{m}$ and $\alpha > 0$, there exists a constant $\beta > 0$ such that
\begin{equation}
\label{eq-thm2-0}
  \lim\limits_{n\rightarrow \infty}\probab{tw(\rig{n, m, p}) > \beta n}{\rig{n, m, p}} = 1.
\end{equation}  
\end{theorem}

\begin{theorem}
\label{thm-power-law-graph}
Let $\rsg{n, m}$ be the Barab\'{a}si-Albert random graph. For any $m \geq 12$, there is 
a constant $\beta > 0$ such that    
\begin{equation}
\label{eq-thm3-0}
  \lim\limits_{n\rightarrow \infty}\probab{tw(\rsg{n, m}) > \beta n}{\rsg{n, m}} = 1.
\end{equation}
\end{theorem}

\subsection{Technical Contribution}
The approach used in \cite{kloks94} is essentially an application of 
the first-moment method
to the random variable that counts the total number of the \textit{balanced partitions}
$(S, A, B)$ where the size of the separator $S$ is at most $\beta n$ 
(See Section~\ref{sec-thm-1} for the formal definition of a balanced partition.)  It is further commented in \cite{kloks94} that it was not known whether the 1.18 lower bound 
can be improved and that the treewidth of the random graph $\rg{n, m}$ with 
$\frac{1}{2}< \frac{m}{n} < 1$ is unknown. 

Our main contribution in this paper is in the proof of our improved lower bound  
$\frac{m}{n} > 1.073$. We note that a more refined analytical calculation is able to 
improve the lower bound 1.18 in \cite{kloks94} to 1.083. The difficulty lies in bringing 
down the lower bound further from 1.083 to 1.073.  To achieve this, we introduce the notion of
$d$-rigid and balanced partitions $(S, A, B)$ which are maximally balanced 
in the sense that no vertex subset of certain size from the larger part, say $B$, can be moved 
to the smaller one $|A|$ to create a new balanced partition. The motivation is that by considering the expected number of these more restricted partitions, 
we will be able to get a more accurate estimation when applying Markov's inequality\footnote{The idea of restricting the kinds of combinatorial objects to be considered have been used in the study of the threshold for the satisfiability of random CNF formulas and the chromatic number of random graphs\cite{kirousis94threshold,achlioptas99thesis,kirousis09threshold}}. 
 
The difficulty we have to overcome in the case of treewidth is the estimation of the expected 
number of $d$-rigid and balanced partition $(S, A, B)$ in $\rg{n, m}$. To do this, 
an exponentially small upper bound is required on the probability that the induced
subgraph $G[B]$ of the random graph $\rg{n, m}$ doesn't have small-sized tree components.

We managed to obtain such an exponentially small upper bound in a ``conditional" probability space, which is equivalent to the Erd\"{o}-R\'{e}yni random model as far as the size of the treewidth is concerned,  by  using a Hoeffding-Azuma style inequality.
To achieve the best possible  Lipschitz constant in our application of the Hoeffding-Azuma inequality, we used a ``weighted" count on the number of tree components of size up to a fixed constant $d$.  We are not aware of any other application of the Hoeffding-Azuma inequality  
in the study of random discrete structures where this idea of weighted counts is beneficial.       

 
\subsection{Outline of the Paper}
The next section fixes our notation and contains preliminaries. Also discussed in this section
is a variant of the Erd\"{o}-R\'{e}nyi model for random graphs which we will be using in our proofs. Sections 3 - 5 contain the proofs of Theorem \ref{thm-treewidth-bound}, 
Theorem \ref{thm-intersection-graph}, and Theorem \ref{thm-power-law-graph} respectively. The two appendices contain the proof of some necessary lemmas.

\section{Notation and Preliminaries}

Throughout this paper, all logarithms are natural logarithms, i.e., 
to the base $e$. The cardinality of a set $U$ is denoted by $|U|$.   
All graphs are undirected and standard 
terminologies in graph theory \cite{west01} are used.  Given a graph $G(V, E)$ and a vertex
$v \in V$, we use $N(v)$ to denote the set of neighbors of $v$, i.e.,
$$
  N(v) = \{u \in V ~|~  u \not= v \textrm{ and } (u, v) \in E \}.
$$
Given a vertex subset $U$, we use $N(U)$ to denote the neighborhood of $U$, i.e.,
$$
N(U) = \{w \in V \setminus U ~|~ (w, u) \in E \textrm{ for some } u \in U\}.
$$  
The induced subgraph on a subset of vertices $U$ is denoted by $G[U]$.
By a component of a graph, we mean a maximal connected subgraph.  

In the proofs, we will be using the following upper bound on $\binom{n}{\delta n}$ that 
can be derived from Stirling's formula:
\begin{lemma}
\label{lem-stirling}
For any constants $0 < \beta < 1$, 
\begin{equation} 
\binom{n}{\beta n } \leq 
   \frac{\theta}{\sqrt{\beta(1 - \beta)n}} \left(\frac{1}{\beta^{\beta}
     (1 - \beta)^{1 - \beta}}\right)^n. \nonumber
\end{equation} 
where $\theta > 0$ is a constant.
\end{lemma}

We also need the  following three lemmas on the properties of some useful functions.  
The proof of these lemmas are incldued in Appendix 2.

\begin{lemma}
\label{lem-function-0}
On internal $(0, 1)$, the function 
\begin{equation}
\label{eq-function-0}
f(t) = t^t(1-t)^{1 - t} \nonumber
\end{equation}  
attains its minimum at  $t = \frac{1}{2}$ and $\lim\limits_{t \rightarrow 0}f(t) = 1$.  Furthermore, $f(t)$ is decreasing on the interval $(0, \frac{1}{2}]$ and decreasing on the 
interval $[\frac{1}{2}, 1)$.
\end{lemma}

\begin{lemma}
\label{lem-function-1}
Let $r(t)$ is a function defined as
\begin{equation}
\label{eq-function-2}
r(t) = \frac{2t^{2}}{(1 + \epsilon)^2c} \left(\frac{1}{e}\right)^{\frac{4ct}{1 -
            2t(1-t)}}
\end{equation}
where $c > 0$ is a constant. For any $c > 1$ and sufficiently small $\beta > 0$,  
$r(t)$ is decreasing on the interval $[\frac{1 - \beta}{2}, \frac{2}{3}]$.
\end{lemma}

\begin{lemma}
\label{lem-function-2}
Let $g(t)$ be a function defined as
\begin{equation}
\label{eq-function-2}
g(t) = \frac{(1 - 2t + 2t^{2} + 2\beta t)^{c}}{t^{t}(1 - t)^{1 - t}}
\end{equation}
where $c > 1$ and $\beta > 0$ are constants. Then for sufficiently small $\beta$,
$g(t)$ is increasing on $[\frac{1 - \beta}{2}, \frac{2}{3}]$.
\end{lemma}

\subsection{Treewidth and Random Graphs}
\label{subsec-treewidth}
The notion of treewith plays an important role in graph theory and in real world computing.   
Several equivalent definitions of treewidth exist and the one based on \ktrees~ is probably 
the easiest to explain. The graph class of \ktrees~is defined recursively as follows
\cite{kloks94}:
\begin{enumerate}
    \item A clique with k+1 vertices is a \ktree;
    \item Given a \ktree~ $T_n$ with n vertices, a \ktree~with
        $n+1$ vertices is constructed by adding to $T_n$ a new
        vertex and connecting it to a $k$-clique of $T_n$.
\end{enumerate}
A graph is called a \textbf{partial \ktree}~if it is a subgraph of a
\ktree. The treewidth $tw(G)$ of a graph $G$ is the minimum value $k$ such that
$G$ is a partial \ktree.

Since the seminal work of Erd\"{o}s and R\'{e}nyi 
\cite{erdos60}, the theory of random graphs  has been an active research area in  graph theory and combinatorics. The probabilistic behavior of various graph properties such as the connectivity, the colorability, and the size of (connected) components, have been extensively 
studied. The theory of random graphs has also motivated the study 
of the probabilistic properties of random instances of other important combinatorial 
optimization problems, most notably that of the satisfiability of random logic formulas in
conjunctive normal form (CNF).    

We use $\rg{n, m}$ to denote an Erd\"{o}s-R\'{e}nyi random graph \cite{bollobas01} on 
$n$ vertices with $m$ edges selected from the $N = \binom{n}{2}$ possible edges 
uniformly at random and without replacement.  Throughout this paper by ``with high probability",
abbreviated as \whp,  we mean that the probability of the event under consideration
is $1 - o(1)$ as $n$ goes to infinity.   

We will be working with a random graph model $\overline{G}(n, m)$ that is slightly different
from $\rg{n, m}$ in that the $m$ edges are selected independently and uniformly
at random, \textbf{but with replacement}. 
There is a one-to-one correspondence between the random graph $\overline{G}(n, m)$ and the 
product probability space 
$(\overline{\Omega}, \mathcal{A}, \probab{\cdot}{\overline{G}(n, m)})$
defined as follows:
\begin{enumerate} 
\item $\overline{\Omega} = \prod\limits_{i = 1}^{m}\mathcal{E}_i$ where each 
$\mathcal{E}_i$ is the set of all $\binom{n}{2}$ possible edges. This is a finite and discrete  sample space.
\item $\mathcal{A}$ is the $\sigma$-field consisting of all subsets of $\overline{\Omega}$.  
\item The probability measure  $\probab{\cdot}{\overline{G}(n, m)}$ is 
\begin{equation}
\probab{\omega}{\overline{G}(n, m)} = \left(\frac{1}{\binom{n}{2}}\right)^{m}, \ \ 
  \forall \omega\in \overline{\Omega}. 
\end{equation}
\end{enumerate}
A sample point $\omega \in \overline{\Omega}$ is interpreted as an outcome of the 
random experiment that selects $m$ edges independently, uniformly at random with replacement
from the set of all possible edges. Note that the graph corresponding to a sample point $\omega \in \overline{\Omega}$  is actually a multi-graph, i.e., a graph in which parallel edges are allowed.    

It turns out that as far as the property of having a treewidth linear in the number of vertices is concerned, the two random graph models $\overline{G}(n, m)$  and 
$\rg{n, m}$ are equivalent. In fact, the equivalence holds for any monotone increasing combinatorial property in random discrete structures,    
as has been observed in \cite{kirousis94threshold,achlioptas99thesis} and formally proved in \cite{kirousis96tech}. For completeness, we will include in Appendix~\ref{appendix-1} an alternative pure measure-theory style proof.   
\begin{proposition}
\label{prop-equiv}
If there exists a constant $\beta > 0$ such that
$$
\lim\limits_{n\rightarrow \infty}\probab{tw(\overline{G}(n, m)) \geq \beta n}{} = 1,
$$
then
$$
\lim\limits_{n\rightarrow \infty}\probab{tw(\rg{n, m}) \geq \beta n}{} = 1.
$$ 
\end{proposition}
Due to Proposition~\ref{prop-equiv}, we will continue to use the notation $G(n, m)$ 
instead of $\overline{G}(n, m)$ throughout this paper, but with the understanding that the $m$ edges are selected independently and uniformly at random with replacement.

In the rest of the paper, we will always subscript operations such as 
a probability measure $\probab{\cdot}{G(n, m)}$ and a mathematical 
expectation $\expectation{\cdot}{G(n, m)}$ to clear indicate the underlying probability space in which these operations are applied. 
 
In \cite{kloks94}, it proved that the treewidth of an Erd\"{o}s-R\'{e}nyi random 
graph $\rg{n, m}$ is linear in the number of vertices \whp~  if the edge/vertex ratio is
greater than $1.18$.  It is mentioned
in \cite{kloks94} that it was unclear whether the lower bound 1.18 can 
be further improved, and that
the treewidth of a random graph $G(n, m)$ with $\frac{1}{2}< \frac{m}{n} < 1$
is unknown \cite{kloks94}.  The main result of this paper improves
the bound to 1.073. 

\subsection{Random Intersection Graphs}
The intersection model for random graphs was introduced by Karon\'{n}ski, et al. \cite{karonski99cpc}. As one of the motivations, Karon\'{n}ski, et al. discussed 
the application of this model in the average-case analysis of algorithmic problems
in gate matrix circuit design \cite{karonski99cpc}.  Other motivations for the recent interests in random intersection graphs include the possible applications in modeling 
large-scale complex networks arising in wireless communications \cite{nikoletseas08tcs} 
and social networks. 

A random intersection graph $\rig{n, m, p}$ over a vertex set $V$ is defined  
by a universe $M$ and  three parameters $n$ (the number of vertices), 
$m = |M|$, and 
$0 \leq p \leq 1$. Associated with a vertex $v\in V$ is a random subset $S_v \subset M$ formed 
by selecting each element in $M$ independently with probability $p$. A pair of vertices $u$
and $v$ is an edge in $\rig{n, m, p}$ if and only if $S_u \cap S_v \not= \emptyset$.   

An alternative view of $\rig{n, p, m}$ is as follows. Let $(C_1, \cdots, C_m)$ be a set of 
$m$ subsets of vertices. Each $C_i$ is formed independently by including each vertex independently with probability $p$. A pair of vertices $u$ and $v$ is an edge in 
$\rig{n, m, p}$  if and only if some $C_i$ contains both $u$ and $v$. In this sense, a random intersection graph is actually the primal graph of a random hypergraph consisting of 
$m$ hyperedges each of which contains a vertex with probability $p$. 
 
\subsection{Barab\'{a}si-Albert Random Graphs} 
In recent years, there has been growing interests in random models for large-scale communication networks, biological networks, and social networks. A remarkable observation is that 
the degree distribution of these large-scale networks follow a power law, i.e., the fraction
of vertices of a given degree $d$ is proportional to $d^{-\gamma}$ for some constant
$\gamma > 0$. 

The Barab\'{a}si-Albert model for random graphs is proposed in  \cite{albert02complex}          
and has been shown to have a power law degree distribution \cite{bollobas01scalefree}. In addition to the purpose of modeling, it is also hoped that features such as a power-law degree distribution may be exploited  algorithmically and/or mathematically 
to help solve real-world problems defined on these large-scale networks. See, for example, the work and arguments in \cite{cooper05,ferrante08tcs,silvio09stoc, gao09tcs}.     

Following the formal definition given in \cite{bollobas01scalefree},
a Barab\'{a}si-Albert random graph  $\rsg{n, m}$ on a set of $n$ vertices 
$V = \{v_1, \cdots, v_n\}$ is defined by a graph evolution process in which vertices are added to the graph one at a time. In each step, the newly-added vertex is connected to
$m$ existing vertices selected according to 
the \textit{preferential attachment} mechanism, i.e. an existing vertex is selected with probability in proportion to its degree. To be more precisely, let $v_i$ be the vertex to be added and let $G_{i - 1}$ be the graph obtained after vertex $v_{i - 1}$ is added. 
The $m$ neighbors of $v_i$ are selected in $m$ steps. In step $1\leq j\leq m$,  
the probability that 
an existing vertex $w$ is selected as the neighbor of the new vertex $v$ is 
\begin{equation}
\label{eq-preferential-prob}
  \frac{\degree{G_{i - 1}}{w} + d_{w}(j) }{m (i - 1) + 2(j - 1)},
\end{equation}          
where 
\begin{enumerate}
\item $(i - 1)m = \sum\limits_{k \leq i - 1} \degree{G_{i - 1}}{v_k}$ is the total degree
of $G_{i - 1}$,
\item $d_{w}(j)$ is the number of times $w$ has been picked as the neighbor of $v$ in the first $(j - 1)$ trials, and 
\item the term $2(j - 1)$ takes into consideration the increase of the total degree as a result of the first $j - 1$ neighbors. 
\end{enumerate}
One also needs to take care of the initial phase, but that won't have any impact on our analyses. 

\section{Treewidth of Erd\"{o}s-R\'{e}nyi Random Graphs: Proof of Theorem \ref{thm-treewidth-bound}}
\label{sec-thm-1}
In this section, we prove Theorem~\ref{thm-treewidth-bound} to establish 
the lower bound $c^*$ on the edge/vertex ratio $\frac{m}{n}$ such that  
whenever $\frac{m}{n} \geq c^*$,  
the treewidth of an  Erd\"{o}s-R\'{e}nyi random graph $\rg{n, m}$ 
is \whp~greater than $\beta n$ for some constant $\beta > 0$.
To begin with, we introduce the following concept which will be used as a necessary condition 
for a graph to have a treewidth of certain size. The following notion of
balanced $l$-partition was used in \cite{kloks94} to establish the lower bound 1.18.
\begin{definition}[\cite{kloks94}]
\label{rigid_partition_def}
Let $G(V, E)$ be a graph with $|V| = n$. 
Let $\mathbf{W} = (S, A, B)$ be a triple of disjoint vertex subsets such that 
$V = S \cup A\cup B$ and $|S| = l + 1$.  

We say that $\mathbf{W}$ is balanced if 
$\frac{1}{3}(n-l-1)\leq |A|, |B| \leq \frac{2}{3}(n - l - 1)$. 
Without lose of generality, we will always assume that $|B| \geq |A|$.

We say that $\mathbf{W}$ is an $l$-partition if 
$S$ separates $A$ and $B$, i.e., there are no edges between vertices of $A$ and vertices of $B$. 
\end{definition}

The following notion of a $d$-rigid partition plays an important role in establishing
our improved lower bound: 
\begin{definition}
Let $d > 0$ be an integer. A triple $\mathbf{W} = (S, A, B)$ with $|B| > |A| + d$ is said to be $d$-rigid if there is no subset of vertices $U \subset B$ with $|U| \leq d$ that induces
a connected component of $G[B]$.
\end{definition}

A $d$-rigid and balanced $l$-partition generalizes Kloks's balanced $l$-partition by
requiring that any vertex set of size at most $d$ in the larger subset of a partition cannot be
moved to the other subset of the partition, and hence the word ``rigid". As we will have to consider all the vertex sets of size at most $d$ to get the best possible estimation, 
the requirement of connectivity is a kind of ``maximality" condition to avoid repeated counting 
of vertex sets of different sizes. For the case of 
$d = 1$, being $d$-rigid means that $G[B]$ has no isolated vertices.

We note that the idea of imposing various restrictions on the combinatorial objects  under consideration has been used in recent years to increase the power of 
the first moment method when dealing with combinatorial problems in 
discrete random structures such as the satisfiability of 
random CNF formulas \cite{kirousis94threshold,kirousis09threshold} and 
the colorability of random graphs \cite{achlioptas99thesis}.      
Our result is a further illustration of the power of this idea in the context 
of treewidth of random graphs.

\begin{lemma}
\label{treewidth-partition-lem}
Let $d \geq 1$  be an integer.
Any graph with a treewidth at most $l > 4$ must have a balanced $l$-partition 
$\mathbf{W} = (S, A, B)$ such that either $|B| \leq |A| + d$ or
$\mathbf{W}$ is $d$-rigid. 
\end{lemma}
\begin{proof}
From \cite{kloks94}, any graph with treewidth at most $l > 4$ must have
a balanced $l$-partition $\mathbf{W} = (S, A, B)$. If $|B| \leq |A| + d$, we are done. Otherwise, if the triple $\mathbf{W}$ is not $d$-rigid, then there must be a vertex subset $U\subset B$ that induces a component of $G[B]$ and consequently 
$$
  N(U) \cap (B\setminus U)  = \emptyset.
$$
Therefore, we can move $U$ from $B$ to $A$ and create a new balanced $l$-partition
with the size of $B$ decreased by $|U|$.
Continuing this process until either $|B| \leq |A| + d$ or 
the partition becomes $d$-rigid. 
\end{proof}

\subsection{Conditional Probability of a $d$-rigid and balanced $l$-partition}

We now bound the conditional probability that a balanced triple $\mathbf{W} = (S, A, B)$
with $|S| = l + 1$ and $|B| \geq |A| + d$ is $d$-rigid given that it is an 
$l$-partition of $\rg{n, m}$.
To facilitate the presentation,  we define the following function
\begin{eqnarray}
\label{eq-rgx-def}
x(t, c) &=& \frac{2ct}{2t^2 - 2t + 1},   \nonumber \\
g(t, c) &=& \sum\limits_{i = 2}^{d}\frac{i^{i - 2}}{i!}
            \left(x(t, c) e^{-x(t, c)}\right)^{i - 1},  \nonumber \\
r(t, c) &=& \frac{2t^2}{(1 + \epsilon)^{2} c} e^{-2x(t, c)}          
\end{eqnarray}
where $\epsilon = \frac{1}{d - 1}$.

\begin{theorem}
\label{thm-rigid-conditional}
Let $G(n, m), c = \frac{m}{n},$ be a random graph and let $\mathbf{W} = (S, A, B)$
be a balanced triple such that $|S| = l + 1, |A| = a, \mbox{ and } |B| = b = tn$. Let $d > 0$
be a constant integer less than $l + 1$.  
Then for $n$ sufficiently large,
\begin{equation}
\label{eq-rigid-probab}
\probab{\mathbf{W} \mbox{ is d-rigid } |\, \mathbf{W} \mbox{ is an l-partition}}{\rg{n, m}}
 \leq \left(\frac{1}{e}\right)^{r(1 + g)^2n}
\end{equation}
where $\epsilon = \frac{1}{d - 1}$,
$$
r = r(t) = \frac{2t^{2}}{(1 + \epsilon)^2c} \left(\frac{1}{e}\right)^{\frac{4ct}{1 -
            2t(1-t)}}
$$
and 
$$
g = g(t) =  \sum\limits_{i = 2}^{d}\frac{i^{i - 2}c^{i - 1}}{i!} 
          e^{-\frac{2(i-1)ct}{2t^2 + 2t - 1}}.
$$

\end{theorem}
\begin{proof} 
Conditional on that $\mathbf{W}$ is an $l$-partition of $\rg{n, m}$, each of the
$m$ edges can only be selected from the set of edges
$$
   E_{W} = V^{2} \setminus \{(u, v): u \in A, v \in B\},
$$
where $V^{2}$ denotes the set of unordered pair of vertices. Let $s$ be the size of 
$E_W$, we have
\begin{eqnarray*}
  s &=& |E_{W}| = \frac{n(n-1)}{2} - ba \\
    &=& \frac{n(n - 1)}{2} - tn(n - tn - (l + 1)). 
\end{eqnarray*}

In the rest of the proof, we will work on the conditional probability space
$\mathcal{P} = (\Omega, \probab{\cdot}{\mathcal{P}})$ where
$\Omega = \Omega_1 \times \Omega_2 \times \cdots \times \Omega_{m}$ and 
$\Omega_i = E_{W}$ for each $1\leq i \leq m$. A sample point 
$\omega = (\omega_1, \cdots, \omega_{m}) \in \Omega$ corresponds to an outcome 
of selecting $m$ edges from $E_W$ uniformly at random and with replacement such that
$\mathbf{W}$ is a balanced $l$-partition of the graph determined by $\omega$. 
The probability measure $\probab{\cdot}{\mathcal{P}}$ is
$$
\probab{\omega}{\mathcal{P}} = \left(\frac{1}{s}\right)^{m}.
$$    
The following lemma guarantees that we can obtain Equation~(\ref{eq-rigid-probab})
by studying the probability $\probab{\mathbf{W} \textrm{ is d-rigid}}{\mathcal{P}}$:  

\begin{lemma}
\label{lem-equiv-condi}
$$
\probab{\mathbf{W} \textrm{ is d-rigid }~ |~ \mathbf{W} \textrm{ is an l-partition}}
  {\rg{n, m}} = \probab{\mathbf{W} \textrm{ is d-rigid}}{\mathcal{P}}.
$$
\end{lemma}
\begin{proof}
Recall that $\probab{\cdot}{\rg{n, m}}$ is the probability measure for the probability 
space $(\overline{\Omega}, \probab{\cdot}{\rg{n, m}})$ and 
$\probab{\cdot}{\mathcal{P}}$ is the probability measure 
for the probability space $\mathcal{P} = (\Omega, \probab{\cdot}{\mathcal{P}})$.
Note that $\Omega$ is the set of sample points $\omega$ in $\overline{\Omega}$ such that 
$\textbf{W}$ is an $l$-partition in the graph determined by $\omega$.
Let $Q \subset \overline{\Omega}$ be the set of sample points $\omega$ such that  
$\mathbf{W}$ is $d$-rigid in the graph determined by $\omega$.
We have 
\begin{eqnarray}
&&\probab{\mathbf{W} \textrm{ is d-rigid }~ |~ \mathbf{W} \textrm{ is an l-partition}}
   {\rg{n, m}} \nonumber \\ 
&&\ \ \ \ \ = 
      \frac{\sum\limits_{\omega \in Q \cap \Omega}\probab{\omega}{\rg{n,m}}}
           {\sum\limits_{\omega \in \Omega}\probab{\omega}{\rg{n,m}}} 
       \ \ \  (\textrm{definition of conditional probability})              
           \nonumber \\ 
&&\ \ \ \ \ = \frac{|Q \cap \Omega|}{|\Omega|}  
 =\frac{|Q \cap \Omega|}{s^m} 
   \nonumber \\
&&\ \ \ \ \ = \sum\limits_{\omega \in Q\cap \Omega}\probab{\omega}{\mathcal{P}} 
     \ \ \  (\textrm{definition of the two probability spaces}) \nonumber \\
&&\ \ \ \ \ = \probab{\mathbf{W} \textrm{ is d-rigid}}{\mathcal{P}}. \nonumber 
\end{eqnarray}
This proves the lemma. 
\end{proof}

Continuing the proof of Theorem~\ref{thm-rigid-conditional}, we need to bound 
$\probab{\mathbf{W} \textrm{ is d-rigid}}{\mathcal{P}}$. 
To make thing simpler, we will 
bound the probability that there exist tree components, instead of general connected components, of size at most $d$ in the subgraph of $\rg{n, m}$ induced on the vertex set $B$.
We use the following variate of Hoeffding-Azuma inequality: 
\begin{lemma}[Lemma 1.2 \cite{mcdiarmid89} and Theorem 1.19 \cite{bollobas01}]
\label{lem-bounded-diff}
Let $\Omega = \prod\limits_{i=1}^{m}\Omega_i$ be a independent 
product probability space where each $\Omega_i$ is a finite set, and $f:\Omega \rightarrow R$
be a random variable satisfying the following Lipschitz condition
\begin{equation}
\label{eq-bounded-diff}
|f(\omega) - f(\omega')| \leq c_f
\end{equation}
if $\omega, \omega' \in \Omega$ differs only in one coordinate. Then, for any $t > 0$,
\begin{equation}
\probab{f(\omega) \leq \expectation{f(\omega)}{} - t }{}
  \leq e^{-\frac{2t^2}{c_f^2m}}. \nonumber
\end{equation}   
\end{lemma} 

In our case, the probability space is $\mathcal{P} = (\Omega, \probab{\cdot}{\mathcal{P}})$
and we may use any the function $f:\Omega \rightarrow R$ such that  the total number 
of tree components of size at most $d$ is larger than zero whenever $f > 0$. To achieve  
the best possible Lipschitz constant $c_f $ in Equation (\ref{eq-bounded-diff}), we consider
a weighted sum $I$ of all tree components of size at most $d$ defined as follows.      

For any $1\leq i\leq d$,  let $\mathcal{U}_i = 
\{U \subset B: |U| = i\}$ be the collection of size-$i$ vertex sets in $B$ 
and let 
$$
\mathcal{U} = \bigcup\limits_{i = 1}^{d} \mathcal{U}_i.
$$ 
For a vertex set $U \in \mathcal{U}$, we use $I_{U}$ to denote the indicator function of the event that  $G[U]$ is a tree component of $G[B]$, i.e., $G[U]$ is a tree and
$N(U) \cap (B \setminus U) = \emptyset$. Define 
\begin{equation}
I = \sum\limits_{U \in \mathcal{U}} \left(1 - (|U| - 1)\epsilon\right)I_{U}
  = \sum\limits_{i = 1}^{d}\sum\limits_{U \in \mathcal{U}_i}
         (1 - (i - 1)\epsilon)I_{U}
\end{equation}  
where $\epsilon = \frac{1}{d - 1}$. The idea is that instead of counting the total number of tree components of size at most $d$, we use the random  variable $I$ as  a ``weighted count" 
to which the contribution of a tree component on a vertex set of size $i$ is $(1 - (i - 1)\epsilon)$. Note that the constant $\epsilon$ can be made arbitrarily small by taking an arbitrarily large (but constant) $d$. The purpose is to make $|I(\omega) - I(\omega')|$  
as close to 1 as possible for every pair $\omega$ and $\omega'$ that differs only on 
one coordinate. 

It is obvious that $I > 0$ if and only if  that the total number of 
tree components of size at most $d$
is greater than zero. By the definition of a $d$-rigid triple, we have
$$
 \probab{\mathbf{W} \textrm{ is d-rigid}}{\mathcal{P}}
   \leq \probab{I = 0}{\mathcal{P}}.
$$  
By Lemma~\ref{lem-equiv-condi} and Lemma~\ref{lem-bounded-diff}, we have 
\begin{eqnarray}
\label{eq-bounded-difference}
&&\probab{\mathbf{W} \mbox{ is d-rigid } |\, \mathbf{W} \mbox{ is an l-partition}}{\rg{n,m}}
       \nonumber \\
&&\ \ \ = \probab{ I = 0 ~|\, \mathbf{W} \mbox{ is an l-partition }}{\mathcal{P}} \nonumber \\
&&\ \ \ \leq \probab{ I - \expectation{I}{\mathcal{P}} \leq 
    - \expectation{I}{\mathcal{P}}} {\mathcal{P}} \nonumber \\
&&\ \ \ \leq \left(\frac{1}{e}\right)^{\frac{2\mathcal{E}^{2}[I]}{c_f^2cn}}
\end{eqnarray}
where $c_f = \max |I(\omega) - I(\omega')|$ with the maximum taken over
all pairs of $\omega$ and $\omega'$ in $\Omega$  that differ only on one coordinate.
The following lemma bounds $\max |I(\omega) - I(\omega')|$. 
(Note that if we had used the unweighted sum $I = \sum\limits_{U \in \mathcal{U}}I_{U}$, 
the best we can have is  $\max |I(\omega) - I(\omega')| \leq 2$.)
\begin{lemma}
\label{lem-max-diff}
For any $\omega, \omega' \in \Omega$ that differ only in one coordinate,  
\begin{equation}
 |I(\omega) - I(\omega')| \leq 1 + \epsilon. \nonumber
\end{equation}
\end{lemma}
\begin{proof}
Note that $\omega$ and $\omega'$ represent two possible outcomes of the $m$ independent random experiments that select the $m$ edges of a random graph.   
If $\omega, \omega' \in \Omega$ differ only in one coordinate, say the $i$-th coordinate, 
then the edge sets of the corresponding graphs $G_{\omega}$ and $G_{\omega'}$ only differ
in the $i$-th edge.   

Let us 
consider the change of the value of $I$ when we modify $G_{\omega}$ to $G_{\omega'}$
by removing the $i$-th edge of $G_{\omega}$ and adding the $i$-th edge of $G_{\omega'}$.
First, removing the $i$-th edge can only increase $I$ by $\delta^+(I)$. The maximum increase
occurs situations where a tree component $T$ is broken up into two smaller
tree components $T_1$ and $T_2$. Suppose that there are $i$ vertices in $T_1$ and $j$ vertices 
in $T_2$, we have 
$$
\delta^{+}_{I} = (1 - (i - 1)\epsilon) + (1 - (j - 1)\epsilon) - 
  (1 - (i + j - 1)\epsilon)I_{i + j \leq d}.
$$
where $i_{i + j \leq d} = 1$ if $i + j \leq d$ and $I_{i + j \leq d} = 0$ otherwise. 
If $i + j \leq d$, we have
$$
\delta^{+}_{I} = (1 - (i - 1)\epsilon) + (1 - (j - 1)\epsilon) - 
  (1 - (i + j - 1)\epsilon) = (1 + \epsilon).
$$ 
If $i + j > d$, we have (since $\epsilon = frac{1}{d - 1}$)
$$
\delta^{+}_{I} = 2 - (i + j - 2)\epsilon  = 2 - (i + j -1)\epsilon + \epsilon < 1 + \epsilon.
$$     
Secondly, adding the $i$-th edge can only decrease $I$ by $\delta^{-}_{I}$. The maximum
decrease occurs in situations where two tree components are merged into a larger one, and
$\delta^{-}_{I} \leq 1 + \epsilon$ as well. 

Therefore, the maximum net change of $I$ is $(1 + \epsilon)$ and is achieved when 
$\delta^{+}_{I} = 1 + \epsilon$ and $\delta^{-}_{I} = 0$, or    
$\delta^{+}_{I} = 0$ and $\delta^{-}_{I} = - (1 + \epsilon)$. Consequently, 
$$
|I(\omega) - I(\omega')| \leq 1 + \epsilon.
$$
The proves the lemma.
\end{proof}

To complete the proof of Theorem~\ref{thm-rigid-conditional}, we estimate in the following 
lemma the expected number of tree components $\expectation{I}{\mathcal{P}}$.   
\begin{lemma}
\label{lem-expected-tree}
Let $I = I(\omega)$ be the number of tree components on at most $d$ vertices in $G[B]$. 
We have 
\begin{equation}
\label{eq-expected-tree}
\expectation{I}{\mathcal{P}} \geq~ 
  te^{-x(t, c)}\left(1 + \sum\limits_{i = 2}^{d}\frac{i^{i - 2}}{i!}
     \left(x(t, c) e^{-x(t, c)}\right)^{i - 1}\right)n  
\end{equation} 
\end{lemma}
\begin{proof}
Let $U,\ |U| = i,$ be a vertex set in $\mathcal{U}_{i}$ and recall that in $\rg{n, m}$, the 
$m = cn$ edges are selected uniformly at random and with replacement. Conditional on the 
event that $\mathbf{W} = (S, A, B)$ is a balanced $l$-partition, the $m$ edges are
selected from the set $E_{W}$ uniformly at random with replacement. Therefore for 
$i \geq 2$, the probability that $G[U]$ is an induced tree component in $G[B]$ is
\begin{eqnarray} 
\probab{I_U = 1}{\mathcal{P}} &=& 
   \binom{cn}{ i - 1} i^{i - 2} \left(\frac{i - 1}{s} \frac{i - 2}{s}
      \cdots \frac{1}{s}\right)  \left(1 - \frac{i (tn - i) +
      \binom{i}{2}}{s}\right)^{cn - i + 1} \nonumber \\
   &=&  c^{i - 1}n^{i - 1}i^{i - 2}\left(\frac{1}{s}\right)^{i - 1}  
       \left(1 - \frac{i (tn - i) + 
      \binom{i}{2}}{s}\right)^{cn - i + 1} \nonumber
\end{eqnarray}
For the case of $|U| = 1$, the probability $\probab{I_U = 1}{\mathcal{P}}$ is the 
probability that the single vertex in $U$ is isolated in $G[B]$, and thus
\begin{equation}
\probab{I_U = 1}{\mathcal{P}} =   
       \left(1 - \frac{(tn - 1)}{s}\right)^{cn}. \nonumber
\end{equation}
Since there are $\binom{tn}{i}$ vertex subsets of size $i$ in $B$, the expected 
number of tree components in $G[B]$ on at most $d$ vertices is
\begin{eqnarray}
\expectation{I}{\mathcal{P}} &=& \sum\limits_{U\in \mathcal{U}} 
      \probab{I_{U} = 1}{\mathcal{P}} \nonumber \\
  &=&  tn \left(1 - \frac{(tn - 1)}{s}\right)^{cn} + \nonumber \\
  &&\ \ \ \ \ \    \sum\limits_{i = 2}^{d}\binom{tn}{i}  
      c^{i - 1}n^{i - 1}i^{i - 2}\left(\frac{1}{s}\right)^{i - 1}  
       \left(1 - \frac{i (tn - i) +  \binom{i}{2}}{s}\right)^{cn - i + 1} \nonumber        
\end{eqnarray}
Since $s = \frac{n(n - 1)}{2} - tn(n - tn - (l + 1)) = \frac{(1- 2t(1-t))n^2 + tn(l + 1) - n}{2}$,
we have that for sufficiently large $n$
\begin{eqnarray}
\expectation{I}{\mathcal{P}} &\geq& tn \left(
   e^{-\frac{2ct}{1 - 2t(1-t)}}
    + \sum\limits_{i = 2}^{d} \frac{t^{i - 1} i^{i - 2} 2^{i - 1}}{(2t^2 - 2t + 1)^{ i - 1} i!} 
     c^{i - 1} e^{-\frac{2ict}{2t^2 - 2t + 1}} \right) \nonumber \\
&=& te^{-x(t, c)}
        \left(1 + \sum\limits_{i = 2}^{d}\frac{i^{i - 2}}{i!} 
          \left(x(t, c)e^{-x(t, c)}\right)^{i - 1}\right)n.  \nonumber      
\end{eqnarray}
This proves Lemma~\ref{lem-expected-tree}.
\end{proof}

To complete the proof of Theorem~\ref{thm-rigid-conditional}, we see that
Equation~(\ref{eq-rigid-probab}) follows from 
Lemma~\ref{lem-max-diff}, Lemma~\ref{lem-expected-tree}, 
and Equation~(\ref{eq-bounded-difference}). 
\end{proof}

\subsection{Proof of Theorem \ref{thm-treewidth-bound}}
We prove Theorem~\ref{thm-treewidth-bound} by applying Markov's inequality and 
the upper bound obtained in Section 3.1 on the conditional probability
of a $d$-rigid and balanced $l$-partition. 

Let $l + 1 = \beta n$ where $\beta > 0$ is a sufficiently small number to be determined
at the end of the proof.  Let $J_1$ be the total number of balanced $\beta n$-partition
$\mathbf{W} = (S, A, B)$ such that $|A| \leq |B| \leq |A| + d$, and  
let $J_2$ be the total number of balanced $\beta n$-partition
$\mathbf{W} = (S, A, B)$ such that $|B| > |A| + d$ and $\mathbf{W}$ is $d$-rigid.    

By Lemma~\ref{treewidth-partition-lem}, if the treewidth of $\rg{n, m}$ is at most $\beta n$, then 
either $J_1 > 0$ or $J_2 > 0$. It follows that
\begin{equation}  
\label{eq-markov}
\probab{tw(\rg{n, m}) \leq \beta n}{\rg{n,m}}
  \leq \probab{J_1 + J_2 > 0 }{\rg{n,m}}. 
\end{equation} 
If we can show that $\expectation{J_1 + J_2}{\rg{n, m}}$ tends to zero  as $n$ goes to infinity, 
Theorem~\ref{thm-treewidth-bound} follows from Markov's inequality. 

Define
\begin{eqnarray}
  \phi_{1}(t) &=& \left(1 - 2t + 2t^{2} + 2t\beta + O(1/n)\right)^{c}, \nonumber \\
 \phi_{2}(t)  &=& \left(e^{-\frac{1}{c}r(t, c)(1 + g(t, c))^2}\right)^{c}, \nonumber \\
  \phi(t) &=& \phi_{1}(t)\phi_{2}(t) \nonumber
\end{eqnarray}
For the expectation of $J_1$, we have   
\begin{lemma}
\label{lem-j-1} 
For any $c > 1$, there is a constant $\beta_1^* > 0$ such that for any $\beta < \beta_1^*$,  
$\lim\limits_{n\rightarrow \infty}\expectation{J_1}{\rg{n,m}} = 0$.
\end{lemma}   
\begin{proof}
Consider a partition $\mathbf{W} = (S, A, B)$ of
the vertices of $\rg{n, m}$ such that $|B| \geq |A|$. Since
$|A| + |B| = (1 - \beta)n$, we see that 
$|B| \leq |A| + d$ if and only if 
$|B| \leq \frac{(1 - \beta)n + d}{2}$. 

The probability that $\mathbf{W}$ is a balanced $l$-partition is
\begin{eqnarray}
\label{treewidth-proof-eq02}
\probab{\mathbf{W} \mbox{ is an } \beta n\textrm{-partition}}{\rg{n,m}} &=&
    \left( 1 - \frac{tn(n - tn - \beta n)}{n(n-1)/ 2}
    \right)^{cn} \nonumber \\
   &=& \left(1 - 2t + 2t^{2} + 2t\beta + O(1/n)\right)^{cn} \nonumber \\
   &=&\phi_1(t).
\end{eqnarray} 
For a fixed vertex set $S$, there are $\binom{(1 - \beta) n}{b} $ ways
($\frac{1}{2}n \leq b = |B| \leq \frac{2}{3}n$) to choose the pair $(A, B)$ such that
one of them has the size $b$. It follows that
\begin{eqnarray}
\expectation{J_1}{\rg{n,m}} &=& \binom{n}{\beta n}
  \sum\limits_{\frac{(1-\beta)n}{2} \leq b \leq \frac{(1 - \beta)n}{2} + d }
           \binom{n - \beta n}{b} \left(\phi_{1}(\frac{b}{n})\right)^{n} \nonumber \\
  &\leq& \binom{n}{\beta n}
  \sum\limits_{\frac{(1-\beta)n}{2} \leq b \leq \frac{(1 - \beta)n}{2} + d }
           \binom{n}{b} \left(\phi_{1}(\frac{b}{n})\right)^{n}. \nonumber             
\end{eqnarray}
Since $\binom{n}{b}$ attains its maximum at $b = \frac{n}{2}$ and the function 
$\phi_{1}(t)$ is increasing in the interval $[\frac{1-\beta}{2}, 1]$, we have
by Stirling's formula (Lemma~\ref{lem-stirling}) that
\begin{eqnarray}
\expectation{J_1}{\rg{n,m}}
  &\leq& d \binom{n}{\beta n}
            \binom{n}{\frac{n}{2}} \left(\phi_{1}(\frac{1}{2})\right)^{n} \nonumber \\
  &\leq& d \binom{n}{\beta n} 2^n (\frac{1}{2} + \beta)^{cn} \nonumber \\
  &\leq& d \left(\frac{1}{\beta^{\beta} (1 - \beta)^{1 - \beta}}\right)^{n} 
    \left(2(\frac{1}{2} + \beta)^c\right)^n.  \nonumber           
\end{eqnarray}            
For any $c > 1$, there is some $\beta_1 > 0$ such that $2(\frac{1}{2} + \beta)^c < 1$
for any $\beta < \beta_1$. Since 
$\lim\limits_{\beta\rightarrow 0}\frac{1}{\beta^{\beta}(1-\beta)^{1 - \beta}} = 1$, there
exists some $\beta_2 > 0$ such that 
$\frac{1}{\beta^{\beta}(1-\beta)^{1 - \beta}} \leq 
  \left(2(\frac{1}{2} + \beta_1)^c\right)^{-1}$.

Taking $\beta^{*} = \min(\beta_1, \beta_2)$, we see that for any $\beta < \beta^{*}$,  
\begin{eqnarray*}
 \expectation{J_1}{\rg{n,m}} &\leq& 
   d \left(\frac{1}{\beta^{\beta} (1 - \beta)^{1 - \beta}}\right)^{n}   
     \left(2(\frac{1}{2} + \beta_1)^c\right)^n \\
 &\leq& d\gamma^{n} 
\end{eqnarray*}
where $0 < \gamma < 1$. Lemma~\ref{lem-j-1} follows. 
\end{proof}   
   
For the expectation of $J_2$, we need to take into consideration the requirement of 
being $d$-rigid in order to get a better bound.  
\begin{lemma}
\label{lem-j-2}
For $c = 1.073$, there is a constant $\beta_2^* > 0$ such that for any $\beta < \beta_2^*$, 
$\lim\limits_{n\rightarrow \infty}\expectation{J_2}{\rg{n,m}} = 0$. 
\end{lemma}    
\begin{proof}
Consider a partition $\mathbf{W} = (S, A, B)$ of
the vertices of $\rg{n, m}$ such that
$|S| = l + 1 = \beta n, |B| \geq |A| + d, |B| = b = tn, \mbox{ with }
\frac{1-\beta}{2} \leq t \leq \frac{2(1 - \beta)}{3}$. Let $I_{\mathbf{W}}$
be the indicator function of the event that
$\mathbf{W}$ is a $d$-rigid and balanced $l$-partition. We have
\begin{eqnarray}
\label{treewidth-proof-eq01}
&&\expectation{I_{\mathbf{W}}}{\rg{n,m}} = 
  \probab{\mathbf{W} \mbox{ is a d-rigid and balanced } \beta n\mathrm{-partition}}{\rg{n,m}}
       \nonumber \\
&&\ \ \ = \probab{\mathbf{W} \mbox{ is a balanced } 
           \beta n\textrm{-partition}}{\rg{n,m}} \times      
            \nonumber \\
&&\ \ \ \ \ \ \ \ \ \   \probab{\mathbf{W} \mbox{ is d-rigid } |\, 
        \mathbf{W} \mbox{ is a balanced }
        \beta n\mathrm{-partition}} {\rg{n,m}}.
\end{eqnarray}
From Theorem \ref{thm-rigid-conditional}, we know that
$$
\probab{\mathbf{W} \mbox{ is d-rigid } |\, \mathbf{W} 
 \mbox{ is a balanced } \beta n\mathrm{-partition}}{\rg{n,m}}
 \leq e^{-r(1 + g)^2n}.
$$
By the definition of a balanced partition,
\begin{eqnarray}
\label{treewidth-proof-eq02}
\probab{\mathbf{W} \mbox{ is a balanced } \beta n\textrm{-partition}}{\rg{n,m}} &=&
    \left( 1 - \frac{tn(n - tn - \beta n)}{n(n-1)/ 2}
    \right)^{cn} \nonumber \\
   &=& \phi_1(t).
\end{eqnarray}
For a fixed vertex set $S$ with $|S| = \beta n$, there are $\binom{n - \beta n}{b} $ ways
($\frac{1}{2}n \leq b \leq \frac{2}{3}n$) to choose the pair $(A, B)$ such that
$|B| = b$. Therefore,
\begin{eqnarray*}
\expectation{J_2}{\rg{n,m}} &=& \sum\limits_{\mathbf{W}} \expectation{I_{\mathbf{W}}}{\rg{n, m}}
            \nonumber \\
&\leq& \binom{n}{\beta n}\sum\limits_{\frac{1}{2}n \leq b \leq \frac{2}{3}n }
           \binom{n - \beta n}{b} \left(\phi(\frac{b}{n})\right)^{n} \nonumber \\
&\leq& \binom{n}{\beta n} \sum\limits_{\frac{1}{2}n \leq b \leq \frac{2}{3}n }
           \binom{n}{b} \left(\phi(\frac{b}{n})\right)^{n}.
\end{eqnarray*}
By Lemma~\ref{lem-stirling}, we have for $n$ large enough
\begin{eqnarray*}
 \expectation{J_2}{\rg{n,m}} &\leq& 
       \left(\frac{1}{\beta^{\beta} (1 - \beta)^{1-\beta}} \right)^{n}
          \sum\limits_{\frac{1}{2}n \leq b \leq \frac{2}{3}n}
           \left(\frac{\phi_{1}(\frac{b}{n})\phi_{2}(\frac{b}{n})}{\frac{b}{n}^{\frac{b}{n}}
           (1 - \frac{b}{n})^{1 - \frac{b}{n}}}\right)^{n}
\end{eqnarray*}
Recall that  
$$
 \phi_{2}(t) = \left(e^{\frac{1}{c}r(t, c)(1 + g(t, c))^2}\right)^{c},
$$
and see Equation (\ref{eq-rgx-def}) for the definition of $r(t, c)$ and $g(t, c)$.
By Lemma~\ref{lem-function-1}, $r(t)$ and $g(t)$ are decreasing on $[\frac{1 - \beta}{2}, \frac{2}{3}]$. Consequently $\phi_2(t)$ is increasing on $[\frac{1 - \beta}{2}, \frac{2}{3}]$.
It follows that
$$
   \phi_{2}(\frac{b}{n}) \leq  \phi_{2}(\frac{2}{3}) 
$$
By Lemma~\ref{lem-function-2},
$$
\frac{\phi_{1}(\frac{b}{n})}{\frac{b}{n}^{\frac{b}{n}} (1 - \frac{b}{n})^{1 - \frac{b}{n}}}
   \leq \frac{\phi_{1}(\frac{2}{3})}{(\frac{2}{3})^{\frac{2}{3}}
           (\frac{1}{3})^{\frac{1}{3}}}
           = \frac{(\frac{5}{9} + \frac{4}{3}\beta)^{c}}
             {(\frac{2}{3})^{\frac{2}{3}}(\frac{1}{3})^{\frac{1}{3}}}.
$$
Therefore,
\begin{equation}
\label{eq-j2-1}
 \expectation{J_2}{\rg{n,m}} \leq 
    O(n)\left(\frac{1}{\beta^{\beta} (1 - \beta)^{1-\beta}} \right)^{n}
   \left(\frac{\left( (\frac{5}{9} + \frac{4}{3}\beta) 
     \phi_2(\frac{2}{3})\right)^{c}}
             {(\frac{2}{3})^{\frac{2}{3}}(\frac{1}{3})^{\frac{1}{3}}}
   \right)^{n}.
\end{equation}
Consider the function 
$$
z(\beta, \epsilon, c) =  \frac{\left((\frac{5}{9} + \frac{4}{3}\beta)
      \phi_2(\frac{2}{3})\right)^{c} }
             {(\frac{2}{3})^{\frac{2}{3}}(\frac{1}{3})^{\frac{1}{3}}}. 
$$
Numerical calculations using MATLAB shows that for $c = 1.073,\ \beta = 0$,\ and $\epsilon = 0$, we have
$$
z(0, 0, 1.073) < 1. 
$$
Since $z(\beta, \epsilon, 1.073)$ is continuous 
in $\beta$ and $\epsilon$ on $[0, 1]$, there exist constants 
$\beta_1 > 0$ and $\epsilon_1 > 0$ such that 
$$
 z(\beta_1, \epsilon, 1.073) < 1, \forall \epsilon < \epsilon_1. 
$$
By Lemma~\ref{eq-function-0}, there exits a constant $\beta_2 > 0$ such that 
$$
 \frac{1}{\beta^{\beta} (1 - \beta)^{1-\beta}} < \frac{1}{z(\beta_1, \epsilon_1, 1.073)}, 
  \forall \beta \leq \beta_2.
$$
Let $\beta^{*} = \min(\beta_1, \beta_2)$. It follows that
 for any $\beta < \beta^*$ and $\epsilon < \epsilon_2$,   
\begin{eqnarray}
  \expectation{J_2}{\rg{n,m}} &\leq&  
       O(n) \frac{1}{\beta_{*}^{\beta_*} (1 - \beta_*)^{1-\beta_*}}
          z(\beta^*, \epsilon^*, 1.073) \nonumber \\
   &\leq& O(n) \frac{1}{\beta_{2}^{\beta_2} (1 - \beta_2)^{1-\beta_2}}
       z(\beta_1, \epsilon, 1.073)    \nonumber \\
   &\leq& O(n) \gamma^{n}        
\end{eqnarray} 
for some constant $0 < \gamma < 1$.
This proves Lemma~\ref{lem-j-2}.
\end{proof}

It follows from Equation (\ref{eq-markov}) that for any $\beta \leq \beta^*$,
$$
 \lim\limits_{n}\probab{tw(G(n, m)) \leq \beta n}{\rg{n,m}} = 0,
\textrm{ if } \frac{m}{n} = 1.073. 
$$
Since the property that the treewidth of a graph is greater $\beta n$
is a monotone increasing graph property, we have that 
for any $c \geq 1.073$,
$$
\lim\limits_{n}\probab{tw(G(n, cn)) \leq \beta n}{\rg{n,m}} = 0.
$$   
Theorem~\ref{thm-treewidth-bound} follows.
\qed

\section{Treewidth of Random Intersection Graphs: Proof of Theorems~\ref{thm-intersection-graph}}
\label{sec-thm-2}


Let $p = \frac{c}{m}$.
Consider a balanced triple $\textbf{W} = (S, A, B)$ with $|S| = \beta n$ and
$|A| = tn$. We upper bound the 
probability that $\textbf{W}$ is a balanced $\beta n$-partition and then use Markov's inequality. By the definition of random intersection graphs, there is no edge between the two vertex sets $A\setminus S$ and $B\setminus S$ if and only if 
\begin{equation}
\label{eq-thm2-1}
e \not\in  \left(\bigcup\limits_{v\in A\setminus S} S_v\right)
      \cap \left(\bigcup\limits_{v\in B\setminus S} S_v\right),\ \forall e \in M,
\end{equation}
which in turn is equivalent to the following: for every $e \in M$, 
\begin{equation}
\label{eq-thm2-2}
\textrm{ either } e \not\in S_v, \forall v \in A\setminus S,\ \  
\textrm{ or } e \not\in S_v, \forall v\in B\setminus S.  
\end{equation}
Since $S_v$'s are formed independently and since $\probab{e\in S_v}{} = p$ for any $e \in M \textrm{ and } v \in V$, the probability for the event in Equation~(\ref{eq-thm2-2}) to occur is
$$
\left((1 - p)^{an} + (1 - p)^{bn} - (1 - p)^{(1 - \beta)n}\right)^m.
$$   
It follows that 
\begin{eqnarray}
&&\probab{\textbf{W} \textrm{ is a balanced } \beta n\textrm{-partition}}{} \nonumber \\  
&&\ \ \ \ \ = \left((1 - p)^{a} + (1 - p)^{b} - (1 - p)^
                    {(1 - \beta)n - a - b}\right)^m \nonumber \\
&&\ \ \ \ \ = (1 - p)^{am} 
  \left(1 + (1 - p)^{(b - a)} - (1 - p)^{(1 - \beta)n - a - b}\right)^{m}. 
\end{eqnarray}
There are $\binom{n}{\beta n}$ ways to choose $S$ and for each fixed $S$, there
are $\binom{n - \beta n}{tn}$ ways to choose $A$ with $|A| = tn$. Since the treewidth of $\rig{n, m, p}$ is at most $\beta n$ implies that there is a balanced $\beta n$-partition, we have by Markov's inequality  that for $p \geq \frac{c}{m}, c > 2$,
\begin{eqnarray*}
&&\probab{tw(\rig{n, m, p}) \leq \beta n}{} \\
&&\ \ \leq \probab{\textrm{There exsits a balanced } \beta n\textrm{-partition} }{}\\
&&\ \ \leq \binom{n}{\beta n}
     \sum\limits_{\frac{1}{3}n \leq a \leq \frac{1}{2}n} \binom{n}{a}
    (1 - p)^{am} \left(1 + (1 - p)^{(b - a)} - (1 - p)^{(1 - \beta)n - a - b}\right)^{m} \\
&&\ \ \leq O(1)\binom{n}{\beta n}
      \sum\limits_{\frac{1}{3}n \leq a \leq \frac{1}{2}n}
         \left( \frac{(\frac{1}{e})^{\frac{a}{n}c}}{(\frac{a}{n})^{\frac{a}{n}} 
             (1 - \frac{a}{n})^{1 - \frac{a}{n}} } \right)^n \\
&&\ \ \leq O(1)n \binom{n}{\beta n} \left(\frac{(\frac{1}{e})^{\frac{2}{3}}}
  {(\frac{1}{3})^{\frac{1}{3}} (\frac{2}{3})^{\frac{2}{3}}} \right)^{n}.                   
\end{eqnarray*}
where last inequality is because the function $\frac{(\frac{1}{e})^{tc}}{t^t(1-t)^{1 - t}}$
is decreasing on $[\frac{1}{3}, \frac{1}{2}]$ for any $c > 2$.
Note that $\frac{(\frac{1}{e})^{\frac{2}{3}}}
  {(\frac{1}{3})^{\frac{1}{3}} (\frac{2}{3})^{\frac{2}{3}}} < 1$.
Therefore, for sufficiently small $\beta$, we have 
$$
\lim\limits_{n\rightarrow \infty}\probab{tw(\rig{n, m, p}) \leq \beta n}{} = 0.
$$
This proves Theorem~\ref{thm-intersection-graph}. 
\qed

\section{The Barab\'{a}si-Albert Model: Proof of Theorem~\ref{thm-power-law-graph}}
\label{sec-thm-3}
Let $V=\{v_1, v_2, \cdots, v_n\}$ be the set of vertices in $\rsg{n, m}$ and
$V_i = \{v_1, \cdots, v_i\}$. Without loss of generality, assume that the vertices are added to $\rsg{n, m}$ in this order in the 
iterative construction of $\rsg{n, m}$.  Let $I_1$ be the first half of the vertices, i.e, $I_1 = \{v_1, v_2, \cdots, v_{\frac{1}{2}n}\}$, and $I_2$
be the second half $\{v_{\frac{1}{2}n + 1}, \cdots, v_n\}$. 

Let $\textbf{W} = (S, A, B)$ be a balanced triple of disjoint vertex subsets 
such that $|S| = \beta n$. (See Definition~\ref{rigid_partition_def} for the details).
Write $|A| = an$ and $|B| = bn$. Assume, without loss of generality, that  $|A| \leq |B|$
so that $\frac{1 - \beta}{3}\leq a\leq \frac{1 - \beta}{2}$.   
Considering the way in which $A$ and $B$ intersect with $I_1$ and $I_2$, let us write
\begin{eqnarray} 
&&|I_1 \cap A| = sn, \ |I_2 \cap A| = (a - s)n; \nonumber \\
&&|I_1 \cap B| = tn, \ |I_2 \cap B| = (b - t)n; 
\end{eqnarray}
where $s$ and $t$ shall satisfy
$$
0\leq s\leq \frac{1 - \beta}{2},\ s + t = \frac{1 - \beta}{2}. 
$$ 

We upper bound the probability 
$\probab{\textbf{W} \textrm{ is a balanced } \beta n\textrm{-partition}}{\rsg{n, m}}$.
Let $E$ be the event that $\textbf{W}$ is a balanced $\beta n$-partition, and
focus on what happens when the second half of the vertices, i.e. those in $I_2$, are added
to $\rsg{n, m}$. Define the following events
\begin{equation}
E_i = \left\{
\begin{array}{l}
  \{ N(v_i) \cap (I_1 \cap B) = \emptyset\}, \textrm{ if } v_i \in I_2\cap A \\
  \{ N(v_i) \cap (I_1 \cap A) = \emptyset\}, \textrm{ if } v_i \in I_2\cap B \\
\end{array}
\right.
\end{equation}
We have
$$
E \subset E_{\frac{n}{2} + 1} \cap \cdots \cap E_{n}.
$$
Therefore,
\begin{equation}
\probab{E}{\rsg{n, m}} \leq \probab{E_{n/2 + 1} \cap \cdots \cap E_{n}}{\rsg{n, m}}.
\end{equation}
The following lemma  bounds the conditional probability of $E_i$ given $\rsg{n, m}[V_{i - 1}]$. 
\begin{lemma}
\begin{equation}
\probab{E_i ~|~ \rsg{n, m}[V_{i - 1}]}{\rsg{n, m}}
  \leq \left\{
\begin{array}{l}
(1 - \frac{s}{2})^m, \textrm{ if } v_i \in I_2 \cap B  \\
\\ 
(1 - \frac{t}{2})^m, \textrm{ if } v_i \in I_2 \cap A
\end{array}  
\right.
\end{equation}
\end{lemma}
\begin{proof}
Consider a vertex $v_i \in I_2 \cap B$ (The case that $v_i \in I_2 \cap A$ is similar).  
The total vertex degree of $\rsg{n, m}[V_{i - 1}]$
is $2(i - 1)m \leq 2nm$. The total vertex degree of the vertices in $I_1 \cap A$ is at least 
$snm$. Note that the event $E_i$ occurs implies that none of the vertices in 
$I_1 \cap A$ is selected as the neighbor of $v_i$ in the $m$-step procedure to pick 
$v_i$'s neighbors.

By the definition of preferential attachment mechanism in the Barab\'{a}si-Albert model, Equation~(\ref{eq-preferential-prob}), we have that 
\begin{eqnarray*}
&&\probab{E_i ~|~ \rsg{n, m}[V_{i - 1}]}{\rsg{n, m}} \\
&&\ \ \ \leq (1 - \frac{snm}{2(i - 1)m}) (1 - \frac{snm}{2(i - 1)m + 2})
      \cdots (1 - \frac{snm}{2(i - 1)m + 2(m - 1)}) \\
&&\ \ \ \leq (1 - \frac{snm}{2nm})^m \\       
&&\ \ \ = (1 - \frac{s}{2})^m.  
\end{eqnarray*}
This proves the lemma.   
\end{proof}

Continue the proof of Theorem~\ref{thm-power-law-graph}. From Lemma,   
we have
\begin{eqnarray}
\probab{E}{\rsg{n, m}} &\leq& \probab{E_{n/2 + 1} \cap \cdots \cap E_{n}}{\rsg{n, m}} 
   \nonumber \\
 &=& \prod\limits_{i = n/2+1}^{n}\probab{E_{i} ~|~ \rsg{n, m}[V_{i - 1}]}{\rsg{n, m}} 
       \nonumber \\
 &\leq& \left((1 - s/2)^{m}\right)^{|I_2 \cap B|}
         \left((1 - t/2)^{m}\right)^{|I_2 \cap A|} \nonumber \\
 &=& \left((1-s/2)^{b - t}(1 - t/2)^{a - s}\right)^{mn}
     \nonumber \\      
\end{eqnarray} 
Taking into consideration that $a + b = (1 - \beta)n$, we see  that
\begin{eqnarray}
\probab{E}{\rsg{n, m}} &\leq& \left((1-s/2)^{b - t}(1 - t/2)^{a - s}\right)^{mn} \nonumber \\
 &=&\left( (1 - s/2)^{b + s - (1 - \beta)/2}
           (3/4 + s/2)^{a - s} \right)^{mn} \nonumber \\
 &=&\left( (1 - \frac{s}{2})^{s - a + \frac{1 - \beta}{2}}
           (\frac{3}{4} +  \frac{s}{2})^{a -s} \right)^{mn}. 
\end{eqnarray}
Consider the behavior of the function 
\begin{eqnarray}
\label{eq-thm3-f-def}
 f(s, \beta) &=& (1 - \frac{s}{2})^{s - a + \frac{1 - \beta}{2}} 
        (\frac{3}{4} +  \frac{s}{2})^{a -s} \nonumber \\
      &=& \left(\frac{1 - s/2}{3/4 + s/2}\right)^{s - a} 
         (1 - s/2)^{\frac{1}{2}} (1 - s/2)^{-\beta/2}.  
\end{eqnarray}
for $0 \leq s \leq \frac{1}{2}$ and $\frac{1 - \beta}{3} \leq a \leq \frac{(1 - \beta)}{2}$. 
We have
\begin{lemma} 
There is a constant $\beta^* > 0$ such that for any $\beta < \beta^*$, 
\begin{equation}
\label{eq-thm3-max}
 f_{max} = \max\{f(s, \beta): s\in[0, 1/2], a \in [(1 - \beta)/3, (1 - \beta)/2]\} < 0.9425.
\end{equation}
\end{lemma}
\begin{proof} 
Note that the last term $(1 - s/2)^{-\beta/2}$ of $f(s, \beta)$ can be made arbitrarily 
to 1 by requiring that $\beta$ is less than a sufficiently small number, say $\beta_0$.  
We, therefore, only need to consider the function 
$$
 f(s) = \left(\frac{1 - s/2}{3/4 + s/2}\right)^{s - a} 
         (1 - s/2)^{\frac{1}{2}}.
$$

First, we claim that $f(s) \leq \left(\frac{7}{8}\right)^{\frac{1}{2}}$
for any $s\in [\frac{1}{4}, \frac{1}{2}]$. 
To see this, we  take the logarithm on both sides of Equation~(\ref{eq-thm3-f-def}) to obtain
$$
\log f(s) = (s - a - \frac{1}{2})\log (1 - \frac{s}{2}) + (a - s)\log (\frac{3}{4} + \frac{s}{2}).
$$
Taking derivative on both sides in the above, we get
$$
\frac{1}{f(s)} f'(s) = \log \frac{1 - \frac{s}{2}}{\frac{3}{4} + \frac{s}{2}}
  - \frac{1}{2} \frac{2s - \frac{7}{4}a + \frac{3}{8}}{(1 - \frac{s}{2})(\frac{3}{4} + \frac{s}{2})}. 
$$
Since for any $s \geq \frac{1}{4}$ and $\frac{1 - \beta}{3} \leq a \leq \frac{1 - \beta}{2}$,  
$\frac{1 - \frac{s}{2}}{\frac{3}{4} + \frac{s}{2}} \leq 1$ and
$2s - \frac{7}{4}a + \frac{3}{8} > 0$, we see that
$f'(s) < 0$. The claim
holds since $f(\frac{1}{4}) = \left(\frac{7}{8}\right)^{\frac{1}{2}} = 0.9354$.

Now consider the interval $[0, \frac{1}{4}]$. Let $\beta_1$ be a constant such that for any 
$\beta < \beta_1$, $\frac{1}{4} - \frac{1 - \beta}{3} < 0$.
Split $[0, \frac{1}{4}]$ into $ d + 1$ segments
and consider the $(d + 1)$ intervals $[s_{i}, s_{i + 1}]$ where
$s_i = i \frac{1}{4d}, \forall 0\leq i\leq d$. Since 
$g(s) = \left(\frac{1 - s/2}{3/4 + s/2}\right)^{s - a}$
is decreasing in $[0, 1/2]$, $s - a < s - \frac{1}{3} < 0$ for any $s \in [0, 1/4]$ and 
$a \in [(1 - \beta)/3, (1 - \beta)/2]$, and $h(s) = (1 - s/2)^{\frac{1}{2}}$ is decreasing in $[0, 1/4]$, we have
\begin{eqnarray*}
\max\limits_{s\in[0, 1/4]}f(s)
   &=& \max_{0\leq i\leq d}\{\max\limits_{s \in [s_i, s_{i + 1}]}f(s)\} \nonumber \\
  &\leq& \max_{0\leq i\leq d}
        (g(s_{i+1}) h(s_i)).
\end{eqnarray*}  
Numerical calculations\footnote{We also tried $d$ up to 10, 000, and found out that the value seems to converge to 0.9424} using $d = 10$ gives us 
$\max_{0\leq i\leq d} (g(s_{i+1}) h(s_i)) < 0.9425$.
Take $\beta^* = \min\{\beta_0, \beta_1\}$, we get Equation~(\ref{eq-thm3-max}).  
\end{proof}

To complete the proof of Theorem~\ref{thm-power-law-graph}, we see from Markov inequality
that the expected number of balanced $\beta$-partition is at most
$$
 \binom{n}{\beta n}\binom{n}{a} \left( (1 - \frac{s}{2})^{s - a + \frac{1}{2}}
           (\frac{3}{4} +  \frac{s}{2})^{a -s} \right)^{mn}
     \leq   \binom{n}{\beta n}\binom{n}{an} 0.9425^{mn}.    
$$ 
Numerical calculation shows that $0.9425^{12} < \frac{1}{2}$. Since $an \leq \frac{1}{2}n$, 
we have by Lemma~\ref{lem-stirling}
and Lemma~\ref{eq-function-0} that there is a constant $\beta_2$ such that for any 
$\beta < \beta_2$, 
$$
\lim\limits_{n\rightarrow \infty}  \binom{n}{\beta n}\binom{n}{an} 0.9425^{mn} = 0.
$$

Let $\beta = \min\{\beta^*, \beta_2\}$ where $\beta^*$ is the constant required 
in Lemma~\ref{eq-thm3-max}. It follows that
for any $m \geq 12$, the expected number of balanced $\beta n$-partitions in 
$\rsg{n, m}$ tends to zero, and consequently 
$$
 \lim\limits_{n\rightarrow \infty}\probab{tw(\rsg{n, m}) > \beta n}{\rsg{n, m}} = 1.
$$ 
This completes the proof of Theorem~\ref{thm-power-law-graph}. 

\appendix 
\section{Proof of Proposition~\ref{prop-equiv}}
\label{appendix-1}
The result actually holds for any monotone increasing combinatorial property in 
random discrete structures,    
as has been observed in \cite{kirousis94threshold,achlioptas99thesis} and formally proved in \cite{kirousis96tech}. For completeness, we give an alternative pure measure-theory style 
proof here.   

Recall that a random graph can be identified with a properly-defined 
probability space. The Erd\"{o}s-R\'{e}nyi random graph $G(n, m)$ corresponds to the probability space $(\Omega_m, \probab{\cdot}{\rg{n, m}})$  where $\Omega$
is the collection of the $\binom{N}{m}$ subsets of $m$ edges ($N = \binom{n}{2}$), and 
$\probab{\cdot}{\rg{n, m}}$ is 
\begin{equation}
\label{eq-app-1-prob-0}
\probab{\omega}{\rg{n, m}} = \frac{1}{\binom{N}{m}}, \forall \omega \in \Omega.
\end{equation}    
Each sample point $\omega\in \Omega_m$ corresponds to a set of $m$ edges selected 
uniformly at random without replacement from the $N$ potential edges.  

The random graph  $\overline{G}(n, m)$, where the $m$ are selected uniformly at random, 
but with replacement, can be identified with the following 
probability space   
$(\overline{\Omega}, \probab{\cdot}{\overline{G}(n, m)})$
where
\begin{enumerate} 
\item $\overline{\Omega}_m = \prod\limits_{i = 1}^{m}\mathcal{E}_i$ where each 
$\mathcal{E}_i$ is the set of all $\binom{n}{2}$ possible edges. 
A sample point 
$\overline{\omega} = \{\overline{\omega}_i, 1\leq i\leq m\} \in \overline{\Omega}_m$
corresponds to a multip-graph with $m$ edges.   
\item The probability measure  $\probab{\cdot}{\overline{G}(n, m)}$ is 
\begin{equation}
\label{eq-prob-space-1}
\probab{\omega}{\overline{G}(n, m)} = \left(\frac{1}{\binom{n}{2}}\right)^{m}.
\end{equation}
Each sample point $\omega\in \overline{\Omega}_m$ is an outcome of the 
random experiment of selecting $m$ edges independently and uniformly at random with replacement
from the set of all possible edges. Also note that the graph represented by 
a sample point in $\overline{\Omega}_m$ is actually a multi-graph, i.e., there are may be
more than one edges between a pair of vertices.  
\end{enumerate}

Let $\beta > 0$ be a fixed constant.
Let $\overline{Q}_m \subset \overline{\Omega}_m$ be the set of sample points 
$\overline{\omega}$ such that the treewidth of the multi-graph determined by 
$\overline{\omega}$ is greater than $\beta n$, and let  
$Q_m \subset \Omega_m$ be the set of sample points $\omega$ such that the treewidth of the simple graph determined by $\omega$ is greater than $\beta n$. 
   
For each $\overline{\omega}\subset \overline{\Omega}_m$, let 
$r(\overline{\omega}) \in \Omega_{|r(\overline{\omega})|}$ be the set of distinct edges 
that $\overline{\omega}$ has, and
let $E_i = \{\overline{\omega}\in \overline{\Omega}_m: |r(\overline{\omega})| = i\}$
be the set of sample points in $\overline{\Omega}$ that have exactly $i$ distinct edges.
For each sample point $\omega \in \Omega_i$, define 
$$
 T_{i}(\omega) = \{\overline{\omega} \in \overline{\Omega}_m: r(\overline{\omega}) = \omega\}. 
$$   
We claim that $\{T_{\omega}: \omega \in \Omega_i\}$  satisfies the following
\begin{equation}
\label{eq-app-1-par-1}
\bigcup\limits_{\omega \in \Omega_i}T_{i}(\omega) = E_i;
\end{equation}
\begin{equation}
\label{eq-app-1-par-2}
T_i(\omega_1) \cap T_i(\omega_2) = \emptyset, \forall 
         \omega_1, \omega_2 \in \Omega_i;
\end{equation}
\begin{equation}
\label{eq-app-1-par-3}
|T_i(\omega_1)| = |T_i(\omega_2)|, \forall 
         \omega_1, \omega_2 \in \Omega_i.    
\end{equation}
If there is an $\overline{\omega}$ that belongs to both $T_i(\omega_1)$ and $T_i(\omega_2)$, 
then it must be the case that $\omega_1 = \omega_2$. Therefore, $T_i(\omega_1) \cap T_i(\omega_2) = \emptyset, \forall \omega_1, \omega_2 \in \Omega_i$. To see that 
$|T_i(\omega_1)| = |T_i(\omega_2)|$, note that any one-to-one mapping $map(\cdot)$ between
the two sets of edges $\omega_1$ and $\omega_2$ defines a one-to-one mapping between
$T_i(\omega_1)$ and $T_i(\omega_2)$. 

From Equations (\ref{eq-app-1-par-1}) through (\ref{eq-app-1-par-3}), 
the additive property of a probability measure, and the fact  that $|\Omega_i| = \binom{N}{i}$,
we see that for any $T_i(\omega)$,
\begin{equation}
\label{eq-app-prob-formula}
\probab{T_i(\omega)}{\overline{G}(n, m)} = \frac{\probab{E_i}{\overline{G}(n, m)}}
    {\binom{N}{i}}.
\end{equation}

Since parallel edges have no impact on treewidth, we have 
\begin{equation}
\label{eq-app-1-par-4}
\textrm{either } T_i(\omega) \cap \overline{Q}_m = \emptyset \textrm{ or } 
  T_i(\omega) \subset \overline{Q}_m,
\end{equation}
and consequently
\begin{equation}
\label{eq-app-1-par-5}
 \bigcup\limits_{\omega \in Q_i}T_i(\omega) = \overline{Q}_m \cap E_i. 
\end{equation} 
We have
\begin{eqnarray}
\probab{\overline{Q}_m}{\overline{G}(n, m)} &=& \sum\limits_{i = 1}^{m}
    \probab{\overline{Q}_m \cap E_i} {\overline{G}(n,m)}  \nonumber \\
&=&\sum\limits_{i = 1}^{m} \sum\limits_{\omega \in Q_i} 
      \probab{T_i(\omega)} {\overline{G}(n,m)} 
       \ \ \ (\textrm{due to (\ref{eq-app-1-par-2}), (\ref{eq-app-1-par-4}), 
          and (\ref{eq-app-1-par-5})}) \nonumber \\      
&=& \sum\limits_{i = 1}^{m} |Q_i| \frac{\probab{E_i}{\overline{G}(n, m)}}{\binom{N}{i}} 
         \ \ \ (\textrm{due to (\ref{eq-app-prob-formula})})  \nonumber \\
&=& \sum\limits_{i = 1}^{m} \probab{Q_i}{G(n, i)} \probab{E_i}{\overline{G}(n, m)} 
       \ \ \ (\textrm{due to (\ref{eq-app-1-prob-0})}) \nonumber \\
&\leq& \probab{Q_m}{G(n, m)}\sum\limits_{i = 1}^{m}  \probab{E_i}{\overline{G}(n, m)} 
       \nonumber \\
&=& \probab{Q_m}{G(n, m)}
\end{eqnarray}
where the second last inequality is due to the fact that the graph property represented 
by the set of sample points $Q_i$  is 
monotone increasing and Theorem 2.1 in \cite{bollobas01} on the probability 
of monotone increasing properties in the Erd\"{o}s-R\'{e}nyi random graph $\rg{n, m}$.  
This completes the proof of the proposition.
\qed

\section{Proof of Lemmas~\ref{lem-function-0},~\ref{lem-function-1} and~\ref{lem-function-2}}
\label{appendix-2}

\subsection{Proof of Lemma~\ref{lem-function-0}}
Taking derivative on both sides of 
$$
\log f(t) = t\log t + (1-t) \log (1 - t),  
$$
we see that $f(t)$ is increasing on $(0, \frac{1}{2}]$ and decreasing
on $[\frac{1}{2}, 1)$. The lemma follows.
\qed

\subsection{Proof of Lemma~\ref{lem-function-1}}
To show that the function
$$
r(t) = \frac{2t^{2}}{(1 + \epsilon)^2c} 
  \left(\frac{1}{e}\right)^{\frac{4ct}{1 - 2t(1-t)}}
$$
is decreasing in $t$ on the interval $[\frac{1 - \beta}{2}, \frac{2}{3}]$, we show that
the its derivative $r'(t) < 0, \forall t \in [\frac{1 - \beta}{2}, \frac{2}{3}]$. 
To this end, we take take the derivative of the logarithm of $r(t)$  
$$
    \log(r(t)) = 2\log(t) - \frac{4ct}{1-2t+2t^{2}} - \log((1 + \epsilon)^2c)
$$
to get
$$
    \frac{1}{r(t)} r'(t) = 
      \frac{2(1-2t+2t^{2})^2 -  4c(t-2t^2+2t^{3}) - 4c(-2t^{2} + 4t^{3})}
             {t(1-2t+2t^{2})^2}.
$$
Since $r(t) > 0$ and $t(1-2t+2t^{2})^2 > 0$, 
we only need to show that the numerator of the right-hand
side in the above, i.e.,
the function
$$
  h(t) = 2(1-2t+2t^{2})^2 -  4c(t-2t^2+2t^{3}) - 4c(-2t^{2} + 4t^{3}).
$$
is less than zero.
 
Note $h(\frac{1}{2}) = \frac{1}{2} - c < 0$ and
$h(\frac{2}{3}) = \frac{50}{81} - \frac{144}{81}c < 0$. As $h(t)$ is continuous, we have 
that for sufficiently small $\beta > 0$, $h(\frac{1 - \beta}{2}) < 0$ as well.     
It is thus sufficient to show
that $h(t)$ itself is monotone. The first and second derivatives of the function
$h(t)$ are respectively
$$
  h'(t) = 4(-2 + 8t - 12t^2 + 8t^3) - 4c(1 - 8t + 18t^2)
$$
and
$$
  h''(t) = 4[ (8 - 24t + 24t^2) - c(-8 + 36t)].
$$
Note that as a quadratic polynomial,
$h''(t) = 4(24t^2 - (24 + 36c)t + 8(1+c))$ can be shown to be always less
than 0 for any $t \in [\frac{1}{2}, \frac{2}{3}]$.  As $h'(t)$ is continuous and
$h'(\frac{1}{2}) = -4c(1 + \frac{1}{2}) < 0$, we see that for sufficiently small
$\beta > 0$, $h'(\frac{1 - \beta}{2}) < 0$ as well.
It follows that
$h'(t) < 0, \forall t \in [\frac{1 - \beta}{2}, \frac{2}{3}]$. Therefore
$h(t)$ is monotone as required.
\qed

\subsection{Proof of Lemma~\ref{lem-function-2}}
First, since both $1 - 2t + 2t^2 + 2\beta t$  and $\frac{1}{t^t(1 - t)^{1 - t}}$
are increasing on the interval $[\frac{1 - \beta}{2}, \frac{1}{2}]$, we have that
$$
 g(t) = \frac{(1 - 2t + 2t^{2} + 2\beta t)^{c}}{t^{t}(1 - t)^{1 - t}}
$$ 
is increasing on the interval   $[\frac{1 - \beta}{2}, \frac{1}{2}]$.

Focusing now on the interval $[\frac{1}{2}, \frac{2}{3}]$, 
let us consider the  logarithm of the function $g(t)$,
$$
  h(t) = \log g(t) = c\log(1 - 2t + 2t^{2} + 2\beta t) - t\log t - (1 - t)\log (1 - t).
$$
The derivative of $h(t)$ is
$$
  h'(t) = c \frac{-2 + 4t + 2\beta}{1 - 2t + 2t^{2} + 2\beta t} 
      - \log t + \log(1 - t)
$$
and $h'(\frac{1}{2}) \geq 0$. The second-order derivative of $h(t)$ is
\begin{eqnarray*}
 h''(t) &=& \frac{c} { (1 - 2t + 2t^{2} + 2\delta t)^{2}t(1 -t)} \times z(t,
                \delta)
\end{eqnarray*}
where
$$
   z(t, \beta) = 4(1 - 2t + 2t^{2} + 2\beta t)(1 - t)t -
                   (4t - 2 + 2\beta)^{2}(1 - t)t -
                   (1 - 2t + 2t^{2} + 2\beta t)^{2}.
$$
First, assume that $\beta = 0$. On the interval $[\frac{1}{2}, \frac{2}{3}]$,
we have
$$
  (4t - 2 + 2\beta)^{2} \leq (4 \times \frac{2}{3} - 2)^{2} = \frac{4}{9},
$$
$$
   \frac{2}{9} \leq  t(1 - t) \leq \frac{1}{2}(1 - \frac{1}{2}) = \frac{1}{4}
$$
and
$$
 \frac{1}{2} \leq (1 - 2t + 2t^{2} + 2\beta t)^{2} \leq ( 1 - 2\times \frac{2}{3} + 2\times
 (\frac{2}{3})^{2})^{2} = \frac{5}{9}.
$$
It follows that
$$
   z(t, \beta = 0) \geq 4 \times \frac{1}{2} \frac{2}{9} - \frac{1}{9} -
   (\frac{5}{9})^{2} = \frac{2}{81} > 0.
$$
Since the family of functions $z(t, \beta), \beta > 0$ are uniformly
continuous on $[\frac{1}{2}, \frac{2}{3}]$, we have that for small enough
$\beta$, $z(t, \beta) > 0, \forall t \in [\frac{1}{2}, \frac{2}{3}]$. 
It follows that the second-order derivative $h''(t)$ is always greater than zero.
Since $h'(\frac{1}{2}) > 0$, we have that $h'(t) > 0, \forall t \in [\frac{1}{2}, \frac{2}{3}]$. It follows that $h(t)$ is increasing. Consequently, $g(t)$ is also increasing since 
$g(t) > 1,~\forall t \in [\frac{1}{2}, \frac{2}{3}]$, 
\qed

\section*{Acknowledgment}
\noindent
A  weaker version of Theorem \ref{thm-treewidth-bound} was reported  in \cite{yong06treewidth}.
The research is supported by Natural Science and Engineering Research Council of Canada (NSERC) RGPIN 327587-06 and RGPIN 327587-09.

\bibliographystyle{plain}
\bibliography{../treewidth-cocoon06,bib/random_graph,bib/phase_transition,bib/complex_networks}
\end{document}